\documentclass[letterpaper, 10 pt, conference]{ieeeconf}  

\IEEEoverridecommandlockouts                              

\usepackage{graphics} 
\usepackage{epsfig} 

\usepackage{amsfonts}
\usepackage{amsmath,amssymb}
\usepackage{comment}
\newtheorem{problem}{Problem}
\newtheorem{proposition}{Proposition}
\newtheorem{definition}{Definition}
\newtheorem{lemma}{Lemma}
\newtheorem{theorem}{Theorem}

\newcommand{\argmin}{\mathop{\rm arg~min}\limits}

\title{\LARGE \bf
Maximum Hands-off Control without Normality Assumption
}

\author{Takuya Ikeda$^{1}$ and Masaaki Nagahara$^{2}$
\thanks{*This research was supported in part by JSPS KAKENHI Grant Numbers
26120521, 15K14006, and 15H02668.}
\thanks{$^{1}$Takuya Ikeda and $^{2}$Masaaki Nagahara are with Graduate School of Informatics, 
         Kyoto University, Kyoto, 606-8501,  Japan.
        {\tt\small ikeda.t@acs.i.kyoto-u.ac.jp} (T. Ikeda), {\tt\small nagahara@ieee.org} (M. Nagahara)}%
}

\begin{document}

\maketitle
\thispagestyle{empty}
\pagestyle{empty}

\begin{abstract}
Maximum hands-off control is a control that has the minimum $L^0$ norm
among all feasible controls.
It is known that the maximum hands-off (or $L^0$-optimal) control problem is
equivalent to the $L^1$-optimal control under the assumption of normality.
In this article,
we analyze the maximum hands-off control for linear time-invariant systems
without the normality assumption.
For this purpose, we introduce the $L^p$-optimal control with $0<p<1$,
which is a natural relaxation of the $L^0$ problem.
By using this, we investigate the existence and the bang-off-bang property
(i.e. the control takes values of $\pm1$ and $0$)
of the maximum hands-off control.
We then describe a general relation between the maximum hands-off
control and the $L^1$-optimal control.
We also prove the continuity and convexity property of the value function,
which plays an important role to prove the stability when the
(finite-horizon) control is extended to
model predictive control.
\end{abstract}

\section{INTRODUCTION}
\label{sec:intro}
In some situations, the control effort can be dramatically reduced 
by {\em hands-off control}, holding the control value exactly zero over a time interval.
The hands-off control is effective in hybrid/electric vehicles, railway vehicles, and
networked/embedded systems~\cite{NagQueNes14a,NagQueNes16}.

Motivated by these applications, 
recently, a novel control method, called {\em maximum hands-off control},
has been proposed in \cite{NagQueNes13,NagQueNes16}.
The purpose of maximum hands-off control is to maximize the time duration 
where the control value is exactly zero among all feasible controls.
The hands-off property is related to sparsity measured by the $L^0$ norm of a signal, 
defined by the total length of the intervals over which the signal takes non-zero values.
This motivates the use of the cost function in which the control effort is penalized via the $L^0$ norm.
The maximum hands-off control, in other words, seeks the sparsest
(or $L^0$-optimal) control among all feasible controls,
and hence the maximum hands-off control is also called {\em sparse optimal control} or {\em $L^0$-optimal control}.

A mathematical difficulty in the maximum hands-off control is 
the discontinuity and the non-convexity of the $L^0$ cost function.
Hence,
recent works~\cite{NagQueNes16,Ikeda_ascc} have proposed to use
the $L^1$ norm for enhancing sparsity,
as often seen in compressed sensing~\cite{Don,EldKut}.
In~\cite{Ikeda_ascc}, 
under the normality assumption (e.g. the plant model is controllable and the $A$-matrix is nonsingular),
the equivalence is proved between the $L^0$-optimal control and the $L^1$-optimal control.
The continuity and the convexity of the value function
is also proved under the same assumption.

Alternatively, a very recent work~\cite{ChaNagQueMal} has proved 
the existence theorem of the maximum hands-off control without the normality assumption, 
by directly dealing with the maximum hands-off control problem without the aid of 
smooth or convex relaxation.
As the necessary condition,
any maximum hands-off control is also proved to have the {\em bang-off-bang property}.
However, the sufficient condition for a control having the bang-off-bang property 
to be $L^0$-optimal is not obtained.

In the present article,
we examine the maximum hands-off control without the normality assumption,
by introducing the $L^p$-optimal control with $0<p<1$.
As will be described in Section \ref{sec:L0-optimal},
$L^p$-optimal control is a relaxation of the maximum hands-off (i.e. $L^0$-optimal) control.
Indeed, the equivalence holds between the $L^0$-optimal control and the $L^p$-optimal control.
The purpose of this article is not only to prove the existence and the bang-off-bang properties
of the maximum hands-off control,
but also to show a general relation between the maximum hands-off control and the $L^1$-optimal control.
The relation leads to the sufficient and necessary condition for a control 
having the bang-off-bang property to be $L^0$-optimal, 
which is not obtained in the recent works.
Also, it leads the equivalence between the value functions in the $L^0$-optimal and the $L^1$-optimal controls, 
by which we prove the convexity and the continuity of the value function.
This property guarantees the stability when the (finite-horizon) maximum hands-off control is extended to model predictive control, as discussed in \cite{IkeNagOno-soav}.


The remainder of this paper is organized as follows: 
In Section~\ref{sec:math}, we give mathematical preliminaries for our subsequent discussion.
In Section~\ref{sec:L0-optimal}, we define the maximum hands-off control problem,
and investigate it via the $L^p$-optimal control.
We show the existence and the bang-off-bang property of the maximum hands-off control
and the relation between the maximum hands-off control and the $L^1$-optimal control.
Section~\ref{sec:value} confirms the continuity and the convexity of the value function.
Section~\ref{sec:example} presents an example to illustrate the difference 
between the maximum hands-off control and the $L^1$-optimal control,
by showing the existence of an $L^1$-optimal control that is not $L^0$-optimal.
In Section~\ref{sec:conclusion}, we offer concluding remarks.

\section{MATHEMATICAL PRELIMINARIES}
\label{sec:math}
This section reviews basic definitions, facts, and notation that will be used
throughout the paper.

Let $n$ be a positive integer.
For a vector $x\in{\mathbb{R}}^n$
and a scalar $\varepsilon>0$, 
the {\em $\varepsilon$-neighborhood} of $x$ is defined by
\[
 {\mathcal B}(x,\varepsilon)\triangleq\{y\in\mathbb{R}^n: \|y-x\|<\varepsilon\},
\] 
where $\|\cdot\|$ denotes the Euclidean norm in ${\mathbb{R}}^n$.
Let ${\mathcal X}$ be a subset of ${\mathbb{R}}^n$.
A point $x\in {\mathcal X}$ is called an {\em interior point} of ${\mathcal X}$ if there exists $\varepsilon>0$ 
such that ${\mathcal B}(x,\varepsilon)\subset {\mathcal X}$. 
The {\em interior} of ${\mathcal X}$ is the set of all interior points of ${\mathcal X}$, 
and we denote the interior of ${\mathcal X}$ by $\mathrm{int}{\mathcal X}$.
A point $x\in\mathbb{R}^n$ is called an {\em adherent point} of ${\mathcal X}$ 
if ${\mathcal B}(x,\varepsilon)\cap {\mathcal X} \neq\emptyset$ for every $\varepsilon>0$, 
and the {\em closure} of ${\mathcal X}$ is the set of all adherent points of ${\mathcal X}$.
A set ${\mathcal X}\subset{\mathbb{R}}^n$ is said to be {\em closed} if 
${\mathcal X}=\overline{{\mathcal X}}$, where $\overline{{\mathcal X}}$ is the closure of ${\mathcal X}$.
The {\em boundary} of ${\mathcal X}$ is the set of all points in the closure of ${\mathcal X}$, 
not belonging to the interior of ${\mathcal X}$, and we denote the boundary of ${\mathcal X}$ by 
$\partial {\mathcal X}$, i.e., $\partial {\mathcal X}= \overline{{\mathcal X}}-\mathrm{int}{\mathcal X}$, 
where $\mathcal{X}_1-\mathcal{X}_2$ is the set of all points which belong to the set ${\mathcal X}_1$ 
but not to the set ${\mathcal X}_2$. 
A set ${\mathcal X}\subset{\mathbb R}^n$ is said to be {\em convex} if,
for any $x,y\in{\mathcal X}$ and any $\lambda\in[0,1]$,
$(1-\lambda)x+\lambda y$ belongs to ${\mathcal{X}}$.


A real-valued function $f$ defined on a convex set ${\mathcal{C}}\subset{\mathbb{R}^n}$ 
is said to be {\em convex} if 
\[
f\bigl((1-\lambda)x+\lambda y\bigr)\leq(1-\lambda)f(x)+\lambda f(y),
\]
for all $x$, $y\in \mathcal{C}$ and all $\lambda\in(0,1)$.


Let $T>0$. 
For a continuous-time signal $u(t)$ over a time interval $[0, T]$, 
we define its {\em $L^p$ and $L^{\infty}$ norms} respectively by
\[
 \|u\|_{p} \triangleq \bigg\{\int_{0}^{T}|u(t)|^{p} dt\bigg\}^{1/p},\quad
 \|u\|_{\infty} \triangleq \sup_{t\in[0, T]}|u(t)|,
\]
where $p\in(0,\infty)$.
Note that $\|\cdot\|_p$ for $p\in(0,\,1)$ is not a norm but a quasi-norm since it fails the triangle inequality~\cite{KalPecRob}.  
We simply denote the set of all signals with $\|u\|_{p}<\infty$ by $L^p$ instead of $L^p[0,\,T]$. 
We define the $L^0$ norm of a signal $u$ on the interval $[0, T]$ as
\[
  \|u\|_{0}\triangleq m(\{t\in[0, T]: u(t)\neq0\}),
\]
where $m$ is the Lebesgue measure on ${\mathbb{R}}$.
Note that $L^0$ ``norm'' is not a norm since it fails the homogeneity property, that is,
for any non-zero scalar $\alpha$ such that $|\alpha|\neq1$,
we have $\|\alpha u\|_0=\|u\|_0\neq|\alpha| \|u\|_0$
for any $u\neq0$.
The notation $\|\cdot\|_0$ derives from the equation in Proposition \ref{pro:L0=limLp}.

\section{MAXIMUM HANDS-OFF CONTROL PROBLEM}
\label{sec:L0-optimal}
In this paper, we consider a linear time-invariant system represented by
\begin{equation}
\dot{x}(t)=Ax(t)+Bu(t), \quad 0\leq t\leq T,
\label{eq:S}
\end{equation}
where 
$A\in\mathbb{R}^{n\times n}$, $B\in\mathbb{R}^{n\times 1}$, 
and $T>0$ is a fixed final time of control.
We here assume single-input control for simplicity.

For the system \eqref{eq:S}, 
we call a control $u\in L^1$ \emph{feasible} if it steers 
$x(t)$ from a given initial state $x(0)=\xi\in\mathbb{R}^n$ to the origin at time $T$ (i.e., $x(T)=0$)
and satisfies the magnitude constraint
$\|u\|_{\infty}\leq1$.
We denote by ${\mathcal U}(\xi)$ the set of all feasible controls for an initial state $\xi\in{\mathbb{R}}^n$, that is,
\[
 \mathcal{U}(\xi)\triangleq
 \biggl\{u\in L^{1}: \int_{0}^{T} e^{-At}Bu(t)dt= -\xi,\quad \|u\|_{\infty}\leq 1\biggr\}.
\]

The control objective is to obtain a control $u\in{\mathcal{U}(\xi)}$
that has the maximum time duration on which $u(t)$ takes $0$.
In other words, 
we seek the control that has the minimum $L^0$ norm among 
all feasible controls in ${\mathcal{U}}(\xi)$.
This optimal control problem is called the \emph{maximum hands-off control problem}.
This is formulated as follows.
\begin{problem}[maximum hands-off control problem]
\label{prob:L0-optimal}
For a given initial state $\xi\in\mathbb{R}^n$, 
find a feasible control $u\in\mathcal{U}(\xi)$ that minimizes
\[J(u)\triangleq \|u\|_0.\]
\end{problem}
We call the optimal control the {\em maximum hands-off control}.

Note that the cost function $J(u)$ can be rewritten as
\[
  J(u) = \int_0^T \phi_0(u(t))~dt,
\] 
where $\phi_0$ is the $L^0$ kernel function defined by
\begin{equation}
  \phi_0(u) \triangleq \begin{cases} 1, & \text{~if~} u\neq 0,\\ 0, & \text{~if~} u=0. \end{cases}
\label{eq:L0-kernel}
\end{equation}
Obviously, 
the kernel function $\phi_0(u)$ is discontinuous at $u=0$ and non-convex.
Also, the cost function $J(u)$ is non-convex, and it has a strong discontinuity.
Indeed,
for any functions $u\neq0$, $v=0$, and any scalar $\lambda\in(0, 1)$, 
we have $\|\lambda u+(1-\lambda) v\|_0=\|u\|_0$.
On the other hand, we have $\lambda\|u\|_0+(1-\lambda)\|v\|_0=\lambda\|u\|_0<\|u\|_0$.
Although the sequence of constant functions $u_k=1/k$ on $[0, T]$ converges to $0$ uniformly,
$J(u_k)$ takes $T$ for any positive integer $k$, and hence it does not converge to $0$.
In contrast, in this paper,
we will show that the value function is continuous and convex on the domain.

First,
we show the existence and the bang-off-bang property of maximum hands-off control 
via {\em $L^p$-optimal control problem}.

\subsection{$L^p$-Optimal Control}
\label{subsec:Lp-optimal}
Here, 
we examine the $L^p$-optimal control problem,
which is formulated as follows.
\begin{problem}[$L^p$-optimal control problem]
\label{prob:Lp}
For a given initial state $\xi\in\mathbb{R}^n$, 
find a feasible control $u\in\mathcal{U}(\xi)$ that minimizes 
\[
  J_{p}(u)\triangleq\|u\|_{p}^{p},
\]
where $p \in (0,\, 1)$.
\end{problem}

We call the solutions to this problem the {\em $L^p$-optimal control}, 
for which the following proposition is fundamental.
\begin{proposition}  
\label{pro:L0=limLp}
For $f\in L^{1}$, we have
\[
  \|f\|_{0} = \lim_{p\to0+} \|f\|_{p}^{p}.
\]
\end{proposition}
\begin{proof}
See Appendix.
\end{proof}

From now on,
we show the existence and the bang-off-bang property of the $L^p$-optimal control.
Let us define the set of all initial states for which there exist feasible controls,
which is known as the {\em reachable set} at time $T$.

\begin{definition}
For the system \eqref{eq:S}, the reachable set $\mathcal{R}$ at time $T$ is defined by
\[
 \mathcal{R}\triangleq\bigg\{\int_{0}^{T}e^{-At}Bu(t)dt: \|u\|_{\infty}\leq 1\bigg\} \subset   
  \mathbb{R}^n.
\]
\end{definition}

The following lemma states the existence and the bang-off-bang property.
\begin{lemma}  
\label{lem:Lp-optimal}
For each initial state in the reachable set $\mathcal{R}$, 
there exist $L^p$-optimal controls,
and they take only $\pm 1$ and $0$ on the time interval $[0,\,T]$.
\end{lemma}
\begin{proof}
The existence of $L^p$-optimal controls is shown in~\cite{Neu},
and we here prove the bang-off-bang property.

Fix any initial state $\xi\in\mathcal{R}$, 
and take any $L^p$-optimal control $u^{\ast}(t)$ for the initial state $\xi$,
and let $x^{\ast}$ denote the resultant state trajectory according to the control $u^{\ast}$. 

The Hamiltonian function for the $L^p$-optimal control problem is defined as
\begin{equation}
 H(x,\,q,\,u)\triangleq |u|^p + q^{\mathrm{T}}(Ax+Bu),
 \label{eq:Hamiltonian_Lp}
\end{equation} 
where $q\in\mathbb{R}^n$ is the costate vector. 
From Pontryagin's minimum principle~\cite{AthFal}, 
there exists a costate vector $q^{\ast}$ that satisfies:
\begin{align}
  &H(x^{\ast},\,q^{\ast},\,u^{\ast})\leq H(x^{\ast},\,q^{\ast},\,u),\quad \forall u\in \mathcal{U}(\xi),\label{ineq:Hamiltonian_Lp}\\
  &{\dot x^{\ast}(t)}=Ax^{\ast}(t)+Bu^{\ast}(t),\quad {\dot q^{\ast}(t)}=-A^{\mathrm{T}}q^{\ast}(t),\notag\\
  &x^{\ast}(0)=\xi,\quad x^{\ast}(T)=0.\notag
\end{align}
From \eqref{eq:Hamiltonian_Lp} and \eqref{ineq:Hamiltonian_Lp}, 
the $L^p$-optimal control $u^{\ast}$ is given by 
\begin{equation}
 u^{\ast}(t)=\argmin_{|u|\leq1} \, |u|^p + (q^{\ast}(t))^{\mathrm{T}}Bu ,\quad t\in[0,\,T].
\label{eq:argmin}
\end{equation}
Hence, from some elementary computation, we have
\begin{equation}
  u^{\ast}(t)=
    \begin{cases}
    1,                         & \mbox{if } (q^{\ast}(t))^{\mathrm{T}}B<-1,\\
    0,                         & \mbox{if } -1<(q^{\ast}(t))^{\mathrm{T}}B<1,\\
    -1,                       & \mbox{if } 1<(q^{\ast}(t))^{\mathrm{T}}B,\\
    0 \mbox{ or } 1,   & \mbox{if } (q^{\ast}(t))^{\mathrm{T}}B=-1,\\
    -1 \mbox{ or } 0, & \mbox{if } (q^{\ast}(t))^{\mathrm{T}}B=1.
    \end{cases}
\end{equation}
on $[0, T]$.
This means that the $L^p$-optimal control $u^{\ast}(t)$ takes only $\pm 1$ and $0$ on $[0,\,T]$.
\end{proof}

From this lemma,
we can show the equivalence between 
the maximum hands-off control and the $L^p$-optimal control.

\begin{theorem}  
\label{thm:relation_L0-Lp}
Let any initial state $\xi\in\mathcal{R}$ be fixed.
Let $\mathcal{U}_{0}^{\ast}(\xi)$ and $\mathcal{U}_{p}^{\ast}(\xi)$ be the sets 
of all maximum hands-off (i.e. $L^0$-optimal) 
controls and all $L^p$-optimal controls, respectively. 
Then we have 
\begin{equation}
  \mathcal{U}_{0}^{\ast}(\xi)=\mathcal{U}_{p}^{\ast}(\xi).
\label{eq:L0-Lp}
\end{equation}
Furthermore, we have 
\begin{equation}
  \|u_{0}\|_{0}=\|u_{p}\|_{p}^{p}
\label{eq:L0-Lp-value}
\end{equation}
for any $u_{0}\in \mathcal{U}_{0}^{\ast}(\xi)$ and $u_{p}\in \mathcal{U}_{p}^{\ast}(\xi)$.
\end{theorem}
\begin{proof}
From Lemma \ref{lem:Lp-optimal},
we can take any $L^p$-optimal control $u_p(t)$, which takes only $\pm 1$ and $0$ on $[0,\,T]$.
Then we have
\begin{align}
  \begin{aligned}
   \|u_{p}\|_{p}^{p}
   &=\int_{0}^{T}|u_{p}(t)|^p dt
     =\int_{\{t:u_{p}(t)\neq0\}}|u_{p}(t)|dt\\
   &=\int_{\{t:u_{p}(t)\neq0\}}1 dt
     =\|u_{p}\|_{0}.
  \end{aligned}\label{lp_optimal}
\end{align}
For any $u\in\mathcal{U}(\xi)$,
we have
\begin{align}
  \begin{aligned}
    \|u\|_{p}^{p}
    &=\int_{0}^{T}|u(t)|^{p}dt
      =\int_{\{t:u(t)\neq0\}}|u(t)|^{p}dt\\
    &\leq\int_{\{t:u(t)\neq0\}}1 dt
      =\|u\|_{0}.
  \end{aligned}\label{any}
\end{align}
From \eqref{lp_optimal}, \eqref{any} and the optimality of $u_{p}$, 
we have
\[
  \|u_{p}\|_{0} = \|u_{p}\|_{p}^{p} \leq \|u\|_{p}^{p} \leq \|u\|_{0}
\]
for any $u\in \mathcal{U}(\xi)$.
This gives $u_{p}\in \mathcal{U}_{0}^{\ast}(\xi)$, 
and hence $\mathcal{U}_{p}^{\ast}(\xi) \subset \mathcal{U}_{0}^{\ast}(\xi)$.
Therefore the set $\mathcal{U}_{0}^{\ast}(\xi)$ is not empty.

Take any maximum hands-off control $u_{0}\in\mathcal{U}_{0}^{\ast}(\xi)$.
From \eqref{any} and the optimality of $u_0$ and $u_p$, we have
\[
  \|u_{p}\|_{p}^{p} \leq \|u_{0}\|_{p}^{p} \leq \|u_0\|_{0} \leq \|u_p\|_{0} = \|u_p\|_{p}^{p},
\]
which yields
\begin{align}
 \|u_0\|_{p}^{p}=\|u_{p}\|_{p}^{p}, \label{eq:L0 in Lp}\\
 \|u_0\|_0=\|u_p\|_{p}^{p}. \label{eq:L0-val=Lp-val}
\end{align}
Equation \eqref{eq:L0 in Lp} gives $\mathcal{U}_{0}^{\ast}(\xi) \subset \mathcal{U}_{p}^{\ast}(\xi)$, 
and hence \eqref{eq:L0-Lp} follows.
Equation \eqref{eq:L0-val=Lp-val} means just the last statement \eqref{eq:L0-Lp-value}.
\end{proof}

In summary,
the maximum hands-off control is characterised 
as follows:

\begin{theorem}  
\label{thm:sparse-optimal}
For each initial state $\xi\in\mathcal{R}$, there exist maximum hands-off controls,
and they take only $\pm 1$ and $0$ on $[0,\,T]$.
\end{theorem}

From the definition of the reachable set $\mathcal{R}$,
this theorem states that 
the initial state $\xi$ exists in $\mathcal{R}$ if and only if maximum hands-off controls exist.

\subsection{Relation between Maximum Hands-Off Control and $L^1$-Optimal Control}
We briefly review the $L^1$-optimal control problem based on the discussion in \cite[Sec. 6-13]{AthFal},
and confirm the definition of the normality.

In the $L^1$-optimal control problem, 
for a given initial state $\xi$,
we seek the control that has the minimum $L^1$ norm among all feasible controls in ${\mathcal{U}}(\xi)$.
The optimal controls are called $L^1$-optimal controls.

We apply the Pontryagin's minimum principle.
Assume that there exists an $L^1$-optimal control $u^{\ast}(t)$.
Then there exists a vector $q^{\ast}(t)\in\mathbb{R}^n$ on $[0, T]$ such that 
 \begin{align*}
 &u^{\ast}(t)=
  \begin{cases}
  1,    &  \mbox{if } B^{\mathrm{T}} q^{\ast}(t)<-1,\\
  0,    &  \mbox{if } | B^{\mathrm{T}} q^{\ast}(t) |<1,\\
  -1,  &  \mbox{if } B^{\mathrm{T}} q^{\ast}(t)>1.
  \end{cases}\\
  & u^{\ast}(t)\in[0, 1], \quad \mspace{2mu}\mbox{if } B^{\mathrm{T}} q^{\ast}(t)=-1,\\
  & u^{\ast}(t)\in[-1, 0], \,\, \mbox{if } B^{\mathrm{T}} q^{\ast}(t)=1.
\end{align*}

Therefore, 
if $|B^{\mathrm{T}} q^{\ast}(t)|$ is not equal to $1$ at almost everywhere in $[0, T]$,
then the $L^1$-optimal control $u^{\ast}(t)$ can be determined uniquely,
and $u^{\ast}(t)$ takes only $0$ and $\pm 1$ on $[0, T]$.
Then the $L^1$-optimal control problem is called {\em normal}.
\begin{definition}[Normality]
\label{normal}
Define the set
\[\mathcal{I} \triangleq \{t\in[0,\,T]: |B^{\mathrm{T}} q^{\ast}(t)|=1\}.\]
If $m(\mathcal{I})=0$, then the $L^1$-optimal control problem is said to be normal.
\end{definition}

Theorem \ref{thm:sparse-optimal} reveals the general relation between maximum hands-off controls and $L^1$-optimal controls,
which is a generalization of a result in~\cite{Ikeda_ascc}.

\begin{theorem}
\label{thm:relation-L0&L1}
Fix any initial state $\xi\in\mathcal{R}$.
Let $\mathcal{U}_{0}^{\ast}(\xi)$ and $\mathcal{U}_{1}^{\ast}(\xi)$ be the sets 
of all maximum hands-off controls and all $L^1$-optimal controls, respectively. 
Then we have
\begin{equation}
  \mathcal{U}_{0}^{\ast}(\xi) \subset \mathcal{U}_{1}^{\ast}(\xi)
\label{eq:L0-L1-set}
\end{equation}
and 
\begin{equation}
\|u_{0}\|_{0}=\|u_{1}\|_{1}
\label{eq:L0-L1-value}
\end{equation}
for any $u_{0}\in \mathcal{U}_{0}^{\ast}(\xi)$ and $u_{1}\in \mathcal{U}_{1}^{\ast}(\xi)$.

In particular, if the $L^1$-optimal control problem is normal,
then we have $\mathcal{U}_{0}^{\ast}(\xi) = \mathcal{U}_{1}^{\ast}(\xi)$.
\end{theorem}
\begin{proof}
From Theorem \ref{thm:sparse-optimal}, we can take any maximum hands-off control $u_0(t)$,
which takes only $\pm 1$ and $0$ on $[0,\,T]$.
There exist a control $u_1\in\mathcal{U}_{1}^{\ast}(\xi)$ which takes only $\pm 1$ and $0$ on $[0,\,T]$,
even if the $L^1$-optimal control problem fails the normality assumption~\cite{Grim}.
Then we have 
\begin{equation}
 \|u_0\|_0 = \|u_0\|_1,\quad
 \|u_1\|_1 = \|u_1\|_0.
\label{eq:bang-off-bang}
\end{equation}
From the optimality of $u_0$ and $u_1$, we also have
\begin{equation}
 \|u_0\|_0 \leq \|u_1\|_0,\quad
 \|u_1\|_1 \leq \|u_0\|_1.
\label{ineq:optimalities}
\end{equation}
It follows from \eqref{eq:bang-off-bang} and \eqref{ineq:optimalities} that 
\[
  \|u_0\|_1 = \|u_1\|_1,\quad
  \|u_0\|_0 = \|u_1\|_1.
\]
The first equation yields the relation \eqref{eq:L0-L1-set},
and the second equation is \eqref{eq:L0-L1-value}.
Finally, for the case under the normality assumption,
see \cite{Ikeda_ascc}.
\end{proof}

We note that any maximum hands-off control is always an $L^1$-optimal control, 
but the reverse does not necessarily hold.
As seen in the proof,
maximum hands-off controls are just $L^1$-optimal controls having the bang-off-bang property.
More precisely, if the normality assumption fails, then 
there exists an $L^1$-optimal control which is not $L^0$-optimal.
In Section \ref{sec:example}, we give such an example.
In contrast, for every initial state,
these optimal control problems always have the same optimal value.

\section{VALUE FUNCTION}
\label{sec:value}
In this section, we investigate the value function in maximum hands-off control.
The value function of an optimal control problem is defined as
the mapping from initial states to the optimal values.
The value functions in maximum hands-off control and $L^1$-optimal control are defined as
\begin{equation*}
  V(\xi)\triangleq\min_{u\in \mathcal{U}(\xi)} \|u\|_{0},\quad
  V_{1}(\xi)\triangleq\min_{u\in \mathcal{U}(\xi)}\|u\|_{1},
\end{equation*}
for $\xi\in\mathcal{R}$.

Here, 
we prove the continuity and the convexity of the value function $V(\xi)$ on $\mathcal{R}$
based on the discussion given in \cite{IkeNagOno-soav}.
As proved there,
these properties play an important role to prove the stability 
when the maximum hands-off control is extended to model predictive control.

First, let us show the convexity of $V(\xi)$.
\begin{theorem}
\label{thm:convex}
The value function $V(\xi)$ is convex on $\mathcal{R}$.
\end{theorem}
\begin{proof}
It is sufficient to show the convexity of $V_1(\xi)$,
since we have
\[
  V(\xi)=V_1(\xi),\quad \xi\in\mathcal{R}
\]
from Theorem \ref{thm:relation-L0&L1}.
Take any initial states $\xi$, $\eta\in\mathcal{R}$ and any scalar $\lambda\in(0,\,1)$.
Let $u_{\xi}(t)$ and $u_{\eta}(t)$ be 
$L^1$-optimal controls for the initial state $\xi$ and $\eta$, respectively.

Obviously, the following control 
\begin{equation}
 u \triangleq (1-\lambda)u_{\xi}+\lambda u_{\eta}
 \label{eq:u}
\end{equation} 
steers the state from the initial state $(1-\lambda)\xi+\lambda \eta$ 
to the origin at time $T$, and it satisfies $\|u\|_{\infty}\leq1$.
That is,  we have $u\in {\mathcal{U}}\bigl((1-\lambda)\xi+\lambda \eta\bigr)$. 
Therefore we have
\begin{align*}
 V_1(\lambda\xi+(1-\lambda)\eta)
  &\leq \|\lambda u_{\xi}+(1-\lambda)\eta\|_1\\
  &\leq \lambda\|u_{\xi}\|_1 + (1-\lambda)\|u_{\eta}\|_1\\
  &=\lambda V_1(\xi) + (1-\lambda)V_1(\eta),
\end{align*}
which shows the convexity of the value function $V_1$.
\end{proof}

Next, we show the continuity of $V(\xi)$.
For this, we prepare the following lemmas.
\begin{lemma}
\label{lem:closedness}
For any scalar $\alpha\geq 0$,
the following set
\[
  \mathcal{R}_{\alpha} \triangleq 
  \{\xi \in \mathcal{R}: V(\xi) \leq \alpha\}
\]
is closed.
\end{lemma}
\begin{proof}
See the proof of \cite[Lemma 4]{IkeNagOno-soav}.
\end{proof}

\begin{lemma}
\label{lem:boundary}
The reachable set $\mathcal{R}$ is characterized by
\begin{equation}
  \mathcal{R}=\{\xi\in\mathcal{R}:V(\xi)\leq T\}.
\label{eq:R}
\end{equation}
In particular, if the pair $(A,\,B)$ is controllable, then
\begin{equation}
  \partial \mathcal{R}=\{\xi\in\mathcal{R}:V(\xi)=T\}.
\label{eq:partial-R}
\end{equation}
\end{lemma}
\begin{proof}
See the proof of \cite[Lemma 5]{IkeNagOno-soav}.
\end{proof}

From these lemmas, we can show the continuity of the value function $V(\xi)$.
\begin{theorem}
If the pair $(A,\,B)$ is controllable,
then $V(\xi)$ is continuous on $\mathcal{R}$.
\label{continuity}
\end{theorem}
\begin{proof}
The proof is similar to the proof of \cite[Theorem 6]{IkeNagOno-soav}.
\end{proof}

\section{EXAMPLE}
\label{sec:example}
Theorem \ref{thm:relation-L0&L1} states that 
for a system that fails the normality assumption,
there exists a control which is $L^1$-optimal but not $L^0$-optimal.
Here, we give such an example.

Let us consider the double-integral system,
which is modelled by
\begin{equation}
 {\dot x(t)}
 =\begin{pmatrix}0&1\\0&0\end{pmatrix}x(t)
 +\begin{pmatrix}0\\1\end{pmatrix}u(t),\quad 0\leq t\leq T.
\label{eq:double-integral}
\end{equation}
Let the initial condition be
\[(x_1(0),\,x_2(0))=(\xi_1,\,\xi_2),\]
where $x(t)=(x_1(t),\,x_2(t))^{\mathrm{T}}$ and $\xi_{i}\in\mathbb{R}$ for $i=1,\,2$.

From~\cite[Control Law 8-3]{AthFal},
this system fails the normality assumption for the initial states satisfying 
\begin{equation}
  \frac{1}{2}\xi_2^2 < \xi_1,\quad  
  \xi_2<0,\quad 
  -\frac{\xi_2}{2}-\frac{\xi_1}{\xi_2}<T.
\label{eq:initial-state-ex1}
\end{equation}
That is,
the $L^1$-optimal control can not be determined uniquely for these initial states.
The set of all $L^1$-optimal controls consists of all controls such that 
\begin{align*}
  &0\leq u(t) \leq 1,\quad \forall t \in[0,T],\\
  &\int_0^T u(t) dt = -\xi_2,\quad
  \int_0^T \int_0^\theta u(t) dtd\theta = -\xi_1-\xi_2 T.
\end{align*}
The dashed line in Fig. \ref{fig:singular-L0&L1}  shows such an $L^1$-optimal control $u_1(t)$ 
(obtained via numerical optimization)
for the parameters
\begin{equation}
  T=5,\quad \xi_1=1,\quad \xi_2=-1,
\label{eq:parameter-ex1}
\end{equation}
which satisfy the condition \eqref{eq:initial-state-ex1}.
Clearly, this $L^1$-optimal control $u_1(t)$ is not $L^0$-optimal.

On the other hand, 
the following is also one of the $L^1$-optimal controls for the initial state satisfying \eqref{eq:initial-state-ex1}:
\[u_2(t)=
\begin{cases}
0,&  \mbox{if } t \in[0,\,t_1)\cup[t_2,\,T],\\
1,&  \mbox{if } t\in[t_1,\,t_2),
\end{cases}\]
where
\[
  t_1=\frac{\xi_2}{2}-\frac{\xi_1}{\xi_2},\quad
  t_2=-\frac{\xi_2}{2}-\frac{\xi_1}{\xi_2}.
\]
This gives a maximum hands-off control 
since the $L^1$-optimal control that takes only $\pm 1$ and $0$ is 
$L^0$-optimal as seen in the proof of Theorem \ref{thm:relation-L0&L1}.
The solid line in Fig. \ref{fig:singular-L0&L1} shows the control $u_2(t)$
for the parameters given in \eqref{eq:parameter-ex1}.
Obviously, the control $u_2(t)$ has much smaller $L^0$ norm than that of $u_1(t)$.
In other words,
the control $u_1(t)$ is $L^1$-optimal, but it is not $L^0$-optimal.

\begin{figure}[thpb]
      \centering
   \includegraphics[width=\linewidth]{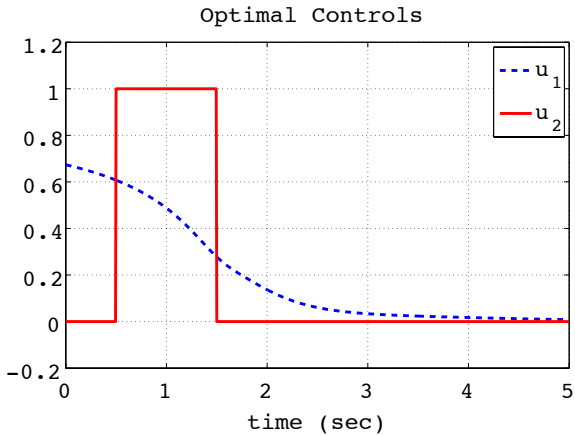}
  \caption{$L^1$-optimal control $u_1$ (dashed) and maximum hands-off control $u_2$ (solid)}
  \label{fig:singular-L0&L1}
\end{figure}

\section{CONCLUSION} 
\label{sec:conclusion}
In this paper, 
we have shown that
among feasible controls
there exists at least one maximum hands-off control,
and it has the bang-off-bang property.
This result is obtained by examining the $L^p$-optimal control
for $0<p<1$,
which is a natural relaxation for the $L^0$-optimal control.
Indeed, we have shown the equivalence between the maximum hands-off control and 
the $L^p$-optimal control.
This leads to the general relation between the maximum hands-off control 
and the $L^1$-optimal control,
that is,
any maximum hands-off control is given by an $L^1$-optimal control 
that has the bang-off-bang property,
but an $L^1$-optimal control is not necessarily $L^0$-optimal
in the absence of the normality assumption.
As an example for this, we have given the double-integral system.
Also we have proved the continuity and the convexity of the value function,
which can be used to prove the stability in model predictive control.

\addtolength{\textheight}{-12cm}   



\section*{APPENDIX}
\section*{PROOF OF PROPOSITION \ref{pro:L0=limLp}}
\label{app:L0=limLp}
For $f\in L^1$,
define
\begin{align*}
  E \triangleq \{t\in[0, T]: |f(t)|\leq1\}, \\
  F \triangleq \{t\in[0, T]: |f(t)|>1\}.
\end{align*}
Then we have
\[
  \|f\|_p^p=\int_{E} |f(t)|^p dt + \int_{F} |f(t)|^p dt.
\]
On the right hand side,
if we take $p\to0+$,
the integrand of the first term increases, and that of the second term decreases.
It follows from Lebesgue's monotone convergence theorem~\cite{Rud} that
\[
 \lim_{p\to0+} \|f\|_{p}^{p} 
  = \int_{E}\phi_{0}(f(t)) dt + \int_{F} \phi_{0}(f(t)) dt
  = \|f\|_0,
\]  
where $\phi_{0}$ is the $L^0$ kernel function defined by \eqref{eq:L0-kernel}.





\begin{thebibliography}{99}

\bibitem{AthFal}
M. Athans and P. L. Falb, 
{\it Optimal Control},
Dover Publications, 1966.

\bibitem{ChaNagQueMal}
D. C. Chatterjee, M. Nagahara, D. Quevedo, and K. S. Mallikarjuna Rao,
Maximum hands-off control: existence and characterization,
submitted for publication.

\bibitem{Don} 
D. L. Donoho, 
Compressed sensing,
{\it IEEE Trans.\ Inf.\ Theory}, 
pp.~1289--1306,  2006.

\bibitem{EldKut} 
Y. C. Eldar and G. Kutyniok, 
{\it Compressed Sensing},
Cambridge University Press, 2012.

\bibitem{Grim}
W. C. Grimmell, 
The existence of piecewise continuous fuel optimal controls,
{\it SIAM J. Control}, vol.~5, no.~4, pp.~515--519, 1967.

\bibitem{Ikeda_ascc}
T. Ikeda and M. Nagahara,
Continuity of the value function in sparse optimal control,
{\it Proc.~of the 10th Asian Control Conference}, 2015.

\bibitem{IkeNagOno-soav}
T. Ikeda, M. Nagahara and S. Ono,
Discrete-valued control by sum-of-absolute-values optimization.
http://arxiv.org/abs/1509.07968, 2015.

\bibitem{KalPecRob}
N. J. Kalton, N. T. Peck, and J. W. Roberts,
{\it An F-Space Sampler},
Cambridge University Press, 1984.

\bibitem{NagQueNes13}
M. Nagahara, D. E. Quevedo, and D. Ne\v{s}i\'{c},
Maximum-hands-off control and $L^1$ optimality,
{\it Proc.~of 52nd IEEE CDC}, Dec.~2013.

\bibitem{NagQueNes14a}
M. Nagahara, D. E. Quevedo, and D. Ne\v{s}i\'{c},
Hands-off control as green control, 
{\it SICE Control Division Multi Symposium}, 2014,
{\tt http://arxiv.org/abs/1407.2377}.

\bibitem{NagQueNes16}
M. Nagahara, D. E. Quevedo, and D. Ne\v{s}i\'c,
Maximum hands-off control: a paradigm of control effort minimization, 
{\it IEEE Trans. Automatic Control}, vol.~61, no.~4, 2016. (to appear)

\bibitem{Neu}
L. W. Neustadt, 
The existence of optimal controls in the absence of convexity conditions,
{\it Journal of Mathematical Analysis and Applications}, 
pp.~110--117, 1963.

\bibitem{Rud} 
W. Rudin, 
{\it Real and Complex Analysis},
3rd ed. McGraw-Hill, New York, 1987.


\end{thebibliography}
\end{document}